\documentclass[12pt]{article}

\usepackage[authoryear]{natbib}
\usepackage{graphicx}
\usepackage{amssymb}
\usepackage{amsmath, amsthm}
\usepackage{times}
\usepackage{bm}

\newcommand{\SBVGBS}{{\rm SBVGBS}}
\newcommand{\GBS}{{\rm GBS}}
\newcommand{\e}{{\rm e}}
\newcommand{\dd}{{\rm d}}
\newcommand{\EL}{{\rm E}}
\newcommand{\SE}{{\rm SE}}
\newcommand{\Es}{\mathbb{E}}
\newcommand{\Var}{\mathbb{V}}

\newcommand{\sign}{{\rm sign}}
\newcommand{\Reais}{\mathbb{R}}

\newcommand{\diag}{{\rm diag}}

\newcommand{\BS}{{\rm BS}}
\newcommand{\SN}{{\rm SN}}
\newcommand{\N}{{\rm N}}
\newcommand{\BVSBS}{{\rm SBVBS}}

\newcommand{\SMVBS}{{\rm SMVBS}}

\newtheorem{theorem}{Theorem}[section]

\title{Skewed Multivariate Birnbaum--Saunders Distributions}
\author{Artur J. Lemonte\\
{\small {\it Departamento de Estat\'istica, Universidade de S\~ao Paulo, S\~ao Paulo/SP, Brazil}}\\
Guillermo Mart\'inez--Florez\\
{\small {\it Departamento de Matem\'aticas, Universidad de C\'ordoba, Monter\'ia, Colombia}} \\
Germ\'an Moreno--Arenas\\
{\small {\it Escuela de Matem\'aticas, Universidad Industrial de Santander, Bucaramanga, Colombia}}
}
\date{}

\begin{document}
\maketitle

\begin{abstract}

The univariate Birnbaum--Saunders distribution has been used quite effectively to model times to
failure for materials subject to fatigue and for modeling lifetime data.
In this article, we define a skewed version of the Birnbaum--Saunders
distribution in the multivariate setting and derive
several of its properties. The proposed skewed multivariate model
is an absolutely continuous distribution whose marginals are univariate
Birnbaum--Saunders distributions. Estimation of the parameters
by maximum likelihood is discussed and the Fisher's information matrix is determined.
A skewed bivariate version for the generalized Birnbaum--Saunders distribution is also
introduced. We provide an application to real data which illustrates the usefulness of the
proposed multivariate model.\\

\end{abstract}

\section{Introduction}\label{introduction}

The univariate family of distributions proposed by \cite{BS1969}, also known as the fatigue life
distributions, has been widely applied for describing fatigue lifetimes.
This family was originally derived from a model for which failure follows from the
development and growth of a dominant crack. A random variable $T$ has
a Birnbaum--Saunders (BS) distribution if it can be written as
$T = \beta\{\alpha Z/2 + [(\alpha Z/2)^2 + 1]^{1/2}\}^2$,
where $Z$ is a random variable following the standard normal distribution,
i.e.~$Z\sim\N(0,1)$.
Its density function is
\begin{equation}  \label{fBS}
f_T(t)=\frac{t^{-3/2}(t+\beta)}{2\sqrt{2\pi}\alpha\sqrt{\beta}}
\exp\biggl[-\frac{1}{2\alpha^2}\biggl(\frac{t}{\beta}+\frac{\beta}{t}-2\biggr)\biggr],\qquad t>0,
\end{equation}
which depends on two parameters: the shape $\alpha>0$ and scale $\beta>0$,
which is also the median of the distribution.
We have $kT\sim\BS(\alpha, k\beta)$ for any $k>0$, i.e.~the BS distribution
is closed under scale transformations.
The expected value, variance, skewness and
kurtosis of $T$ are, respectively,
\[
\Es(T) = \beta\biggl(1 + \frac{1}{2}\alpha^2\biggr),
\qquad
\Var(T) = (\alpha\beta)^2\biggl(1 + \frac{5}{4}\alpha^2\biggr),
\]
\[
\gamma_{3} = \frac{16\alpha^2(11\alpha^2 + 6)}{(5\alpha^2 + 4)^3},
\qquad
\gamma_{4} = 3 + \frac{6\alpha^2(93\alpha^2 + 41)}{(5\alpha^2 + 4)^3}.
\]
The density function \eqref{fBS} is right skewed and the skewness decreases with $\alpha$.
Notice that both mean and variance increase as $\alpha$ increases.
It is also of interest to mention that if $T\sim\BS(\alpha, \beta)$,
then $T^{-1}\sim\BS(\alpha, \beta^{-1})$. It implies that the BS distribution
also belongs to the family of random variables closed under reciprocation
\citep{Saunders1974}. It then follows that
\[
\Es(T^{-1}) = \beta^{-1}\biggl(1 + \frac{1}{2}\alpha^2\biggr),
\qquad
\Var(T^{-1}) = \alpha^2\beta^{-2}\biggl(1 + \frac{5}{4}\alpha^2\biggr).
\]
The shape of the hazard function of the BS distribution
is discussed in \cite{kundu-et-al-2008}. The authors showed
that the hazard rate function is not monotone and
is unimodal for all ranges of the parameter values.
Some interesting results on improved statistical inference for the
BS distribution may be revised in \cite{WuWong2004} and \cite{LCNV07, LSCN08}.

The univariate BS distribution has received significant attention
over the last few years by many researchers and
some generalizations are proposed in \cite{Diaz-Leiva05},
\cite{Owen2006}, \cite{guiraud-et-al-2009}, \cite{Leiva-et-al-2009},
\cite{Castillo-et-al-2011} and \cite{CordeiroLemonte2011}, among other.
On the other hand, as far as we know, little work has been done to extend the
BS distribution to the multivariate case. We can refer to the works by
\cite{DiazGarciaDominguezMolina2006}, \cite{Kundu-et-al-2010} and
\cite{Caro-Lopera-et-al-2012}. In \cite{DiazGarciaDominguezMolina2006},
the authors defined an independent multivariate BS distribution.
By using the bivariate normal distribution function,
\cite{Kundu-et-al-2010} proposed a bivariate BS distribution
which is absolutely continuous and has five parameters.
Finally, \cite{Caro-Lopera-et-al-2012} introduced the
matrix-variate generalized BS distribution.

As can be observed, little work on multivariate versions
for the BS distribution have been published. In this paper,
in addition to the existing multivariate BS models,
we shall propose the asymmetric (skewed) multivariate BS distribution
based on the work of \cite{Arnold-et-al-2002}. The main motivation for introducing
this multivariate version of the BS distribution
relies on the fact that the practitioners will have a new multivariate BS
model to use in multivariate settings, since the formulae related with the new
multivariate model are manageable and with the use of modern computer
resources and its numerical capabilities, the proposed model may prove to be an
useful addition to the arsenal of applied statisticians.
Additionally, the new model is quite flexible (see Figure \ref{densities_plots}
in Section \ref{bivcase}) and can
be widely applied in analyzing multivariate data.
Further, we provide an application to real data in which
is showed that the new multivariate model yields a better
fit than other multivariate BS distributions available in the literature.

The paper unfolds as follows. The skewed bivariate BS distribution
is defined in Section \ref{bivcase} and then several properties are discussed.
The multivariate extension is presented in Section \ref{mult_case}.
In Section \ref{estimation_infer}, we propose different methods for
estimating the unknown parameters as well as derive the
information matrix and discuss likelihood ratio tests for some hypotheses of interest.
In particular, we propose modified moment estimators for the unknown parameters which are
explicit in form and can therefore be used effectively as the
initial guess in the iterative process for the computation of the
maximum likelihood estimators. Further, the asymptotic distribution of the
maximum likelihood estimators is derived and thus the asymptotic confidence
intervals for the unknown parameters can be constructed.
The usefulness of the proposed model is illustrated in
an application to real data in Section \ref{application}.
We also introduce in Section \ref{gBSbiv} the
skewed bivariate generalized BS distribution.
Finally, Section \ref{conclusions}
closes the paper with some concluding remarks.

\section{Skewed bivariate BS distribution}\label{bivcase}

We initially consider the skewed bivariate BS distribution.
For each $x\in\Reais$ and for each $y\in\Reais$, consider the conditional distributions
\begin{equation}\label{cond_dist}
X|Y=y \sim \SN(\lambda y),
\qquad
Y|X=x\sim\SN(\lambda x),
\end{equation}
where $X|Y=y \sim \SN(\lambda y)$ means that given $Y=y$, $X|Y=y$
has skew normal distribution \citep{Azzalini1985}. The shape parameter
$\lambda\in\Reais$ determines the skewness of the density.
From \cite{Arnold-et-al-2002} and using the conditional
distributions in \eqref{cond_dist}, the joint
probability density function (pdf) of the random vector $(X, Y)$ takes the form
\begin{equation}\label{Arnoldskew}
f_{X,Y}(x,y)=2\phi(x)\phi(y)\Phi(\lambda xy),\qquad (x,y)\in\Reais^2,
\end{equation}
where $\phi(\cdot)$ and $\Phi(\cdot)$ denote the pdf
and cumulative distribution function (cdf) of the standard normal distribution,
respectively. Also, $f_X(x)=\phi(x)$ and $f_Y(y)=\phi(y)$. If $\lambda=0$ in \eqref{Arnoldskew}, then
$f_{X,Y}(x,y)=\phi(x)\phi(y)$ and hence $X$ and $Y$ become independent.
For $\lambda\neq 0$, it can be shown that the correlation between $X$ and $Y$,
$\rho(X,Y)$ say, is given by 
\[
\rho(X,Y) = \sign(\lambda)\times \frac{U(3/2,2,1/(2\lambda^2))}{2\lambda^2\sqrt{\pi}},
\]
where $U(a,b,z)$ denotes the confluent hypergeometric function, defined as
\[
U(a,b,z)=\frac{1}{\Gamma(a)}\int_{0}^\infty\e^{-zt}t^{a-1}(1+t)^{b-a-1}\dd t,
\]
with $b>a>0$ and $z>0$, and $\Gamma(\cdot)$ represents the gamma function.
Therefore, the parameter $\lambda$ also governs the correlation.

Let $Z_j \sim\N(0,1)$, for $j=1,2$, with $Z_1|Z_2=z_2\sim\SN(\lambda z_2)$ and
$Z_2|Z_1=z_1\sim\SN(\lambda z_1)$.
Then, taking the transformation
\[
T_j=\beta_{j}\left[\frac{\alpha_{j}}{2}{Z_j}+\sqrt{\left(\frac{\alpha_{j}}{2}{Z_j}\right)^2+1}\right]^2,\qquad j=1,2,
\]
where $\alpha_{j}>0$ and $\beta_{j}>0$, the joint pdf  of the skewed bivariate
BS (SBVBS) distribution takes the form
\begin{equation}\label{fBGBS1}
f_{T_1,T_2}(t_1,t_2)=2\phi(a_{1})\phi(a_{2})\Phi(\lambda a_{1}a_{2})
\frac{t_1^{-3/2}(t_1+\beta_{1})}{2\alpha_{1}\sqrt{\beta_{1}}}
\frac{t_2^{-3/2}(t_2+\beta_{2})}{2\alpha_{2}\sqrt{\beta_{2}}}, \qquad(t_1,t_2)\in\Reais_+^2,
\end{equation}
where
\begin{equation}\label{at1}
a_{j}=a_{j}(\alpha_{j},\beta_{j})=\frac{1}{\alpha_{j}}\biggl[\biggl(\frac{t_{j}}{\beta_{j}}\biggr)^{1/2}
-\biggl(\frac{\beta_{j}}{t_{j}}\biggr)^{1/2}\biggr],\qquad j=1,2.
\end{equation}
The notation used is $(T_1, T_2)\sim\BVSBS(\alpha_{1}, \alpha_{2}, \beta_{1}, \beta_{2}, \lambda)$.
The random variables $T_1$ and $T_2$ become independent for $\lambda=0$ in \eqref{fBGBS1}
and hence the proposed bivariate
model reduces to the independent bivariate model considered by \cite{DiazGarciaDominguezMolina2006}.
So, as remarked, the shape parameter $\lambda$ also introduces
correlation between $T_1$ and $T_2$.

Contour plots for the joint pdf \eqref{fBGBS1} are presented
in Figure \ref{densities_plots}. From this figure, note that \eqref{fBGBS1} can take on
different shapes and will therefore be useful in analyzing
bivariate data. Additionally, notice that \eqref{fBGBS1} can be unimodal or
bimodal depending on the value of $\lambda$.
\begin{figure}[!htp]
\centering
\includegraphics[scale=0.45]{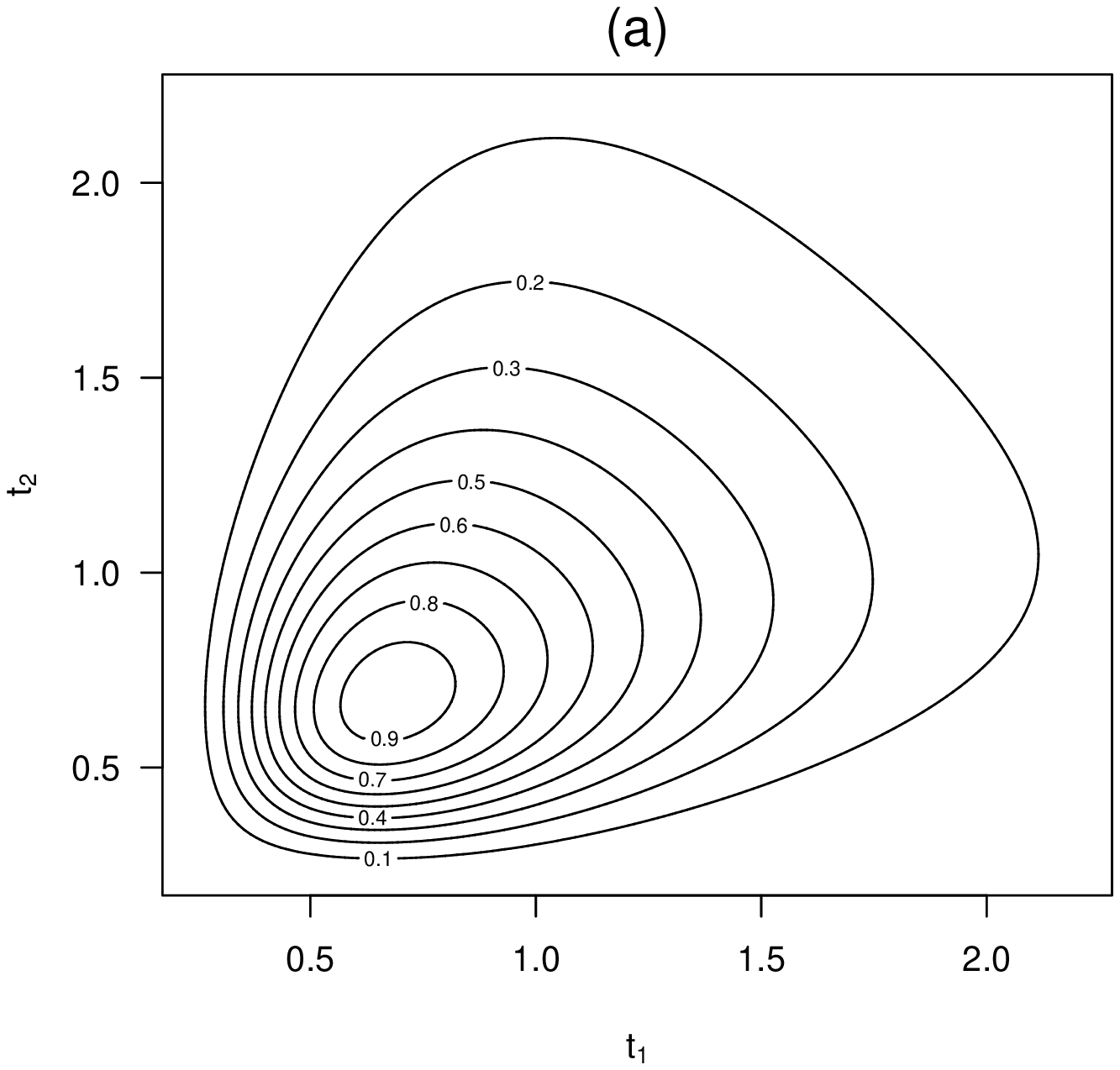}
\includegraphics[scale=0.45]{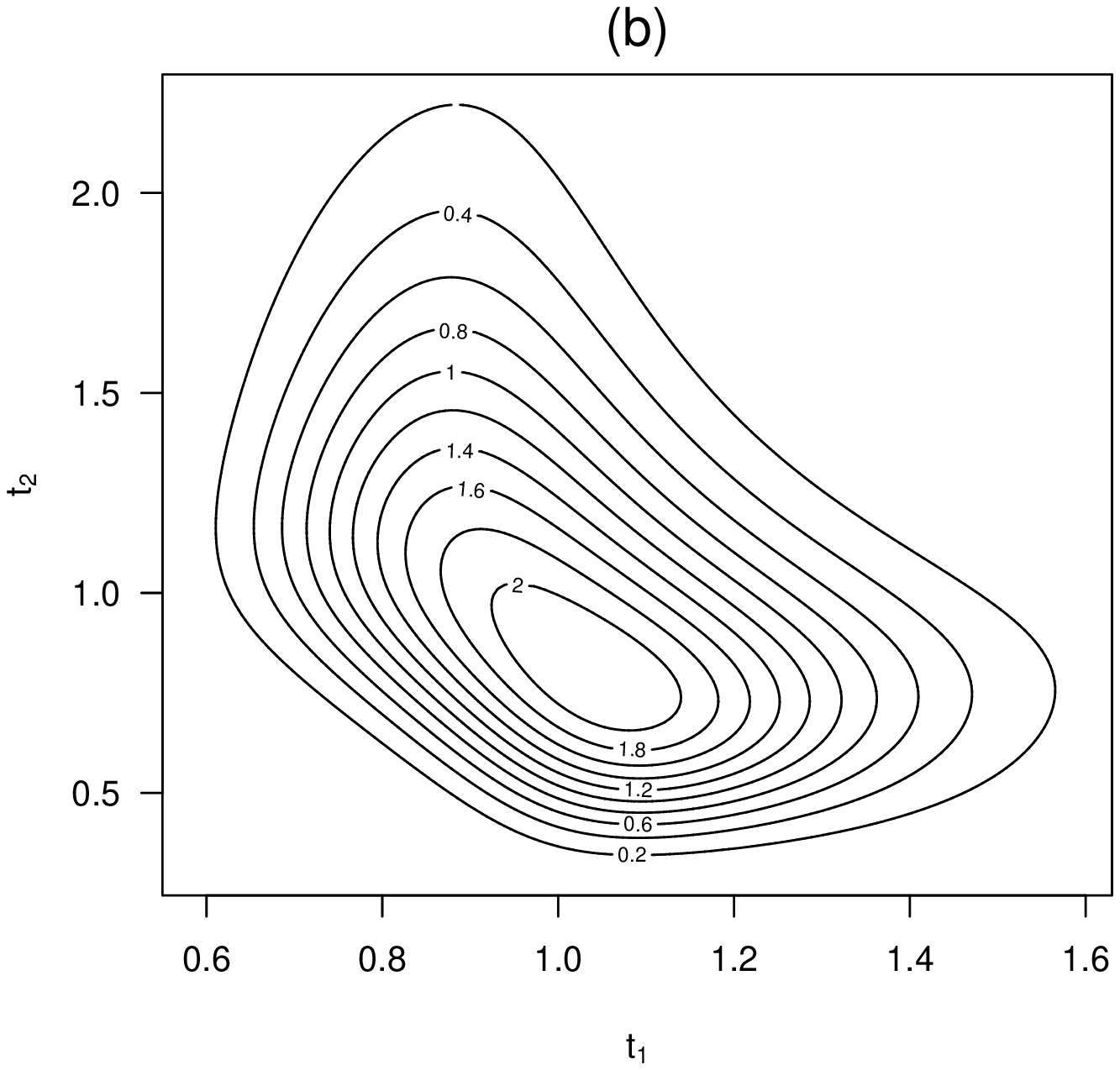}
\includegraphics[scale=0.45]{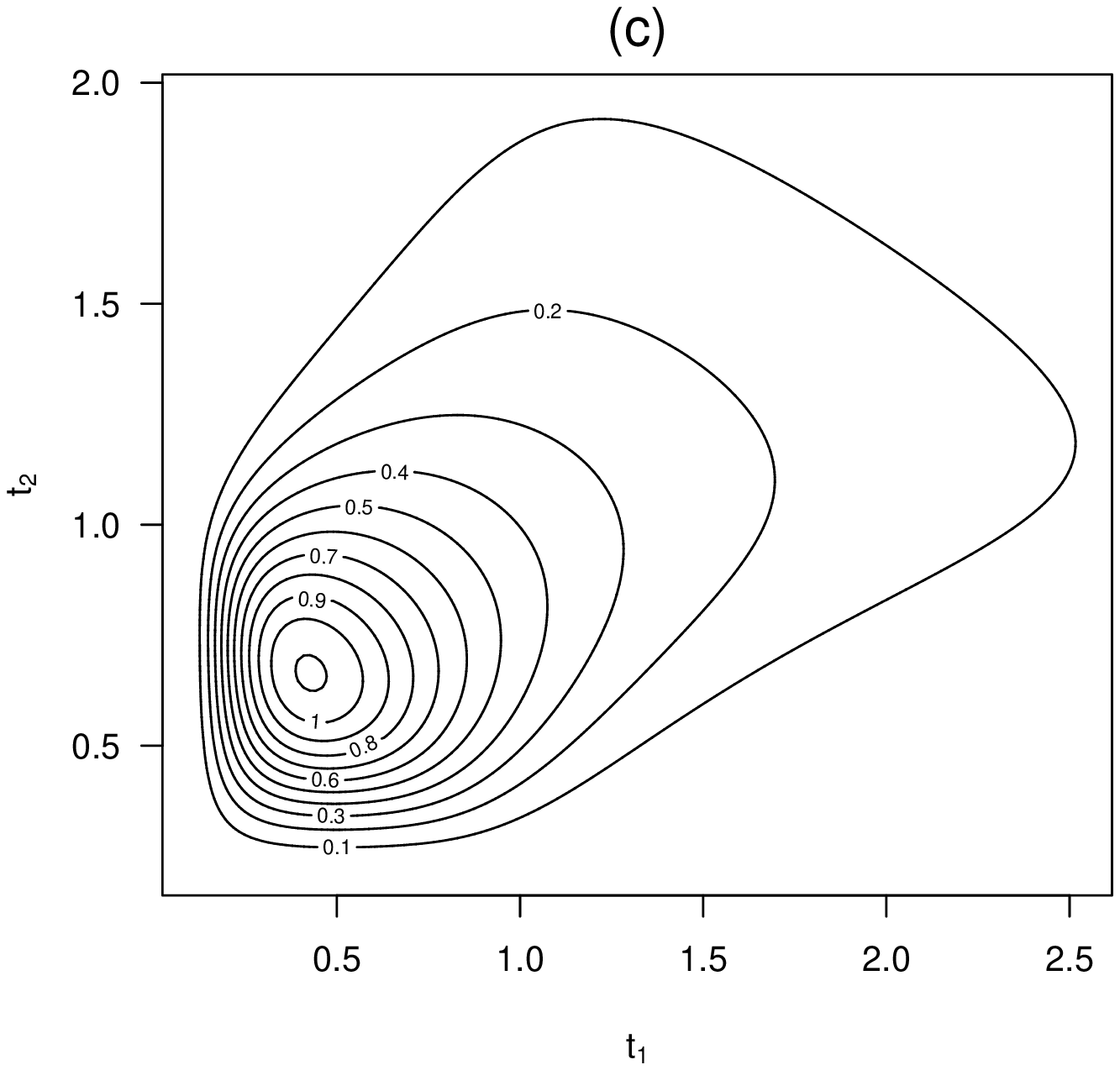}
\includegraphics[scale=0.45]{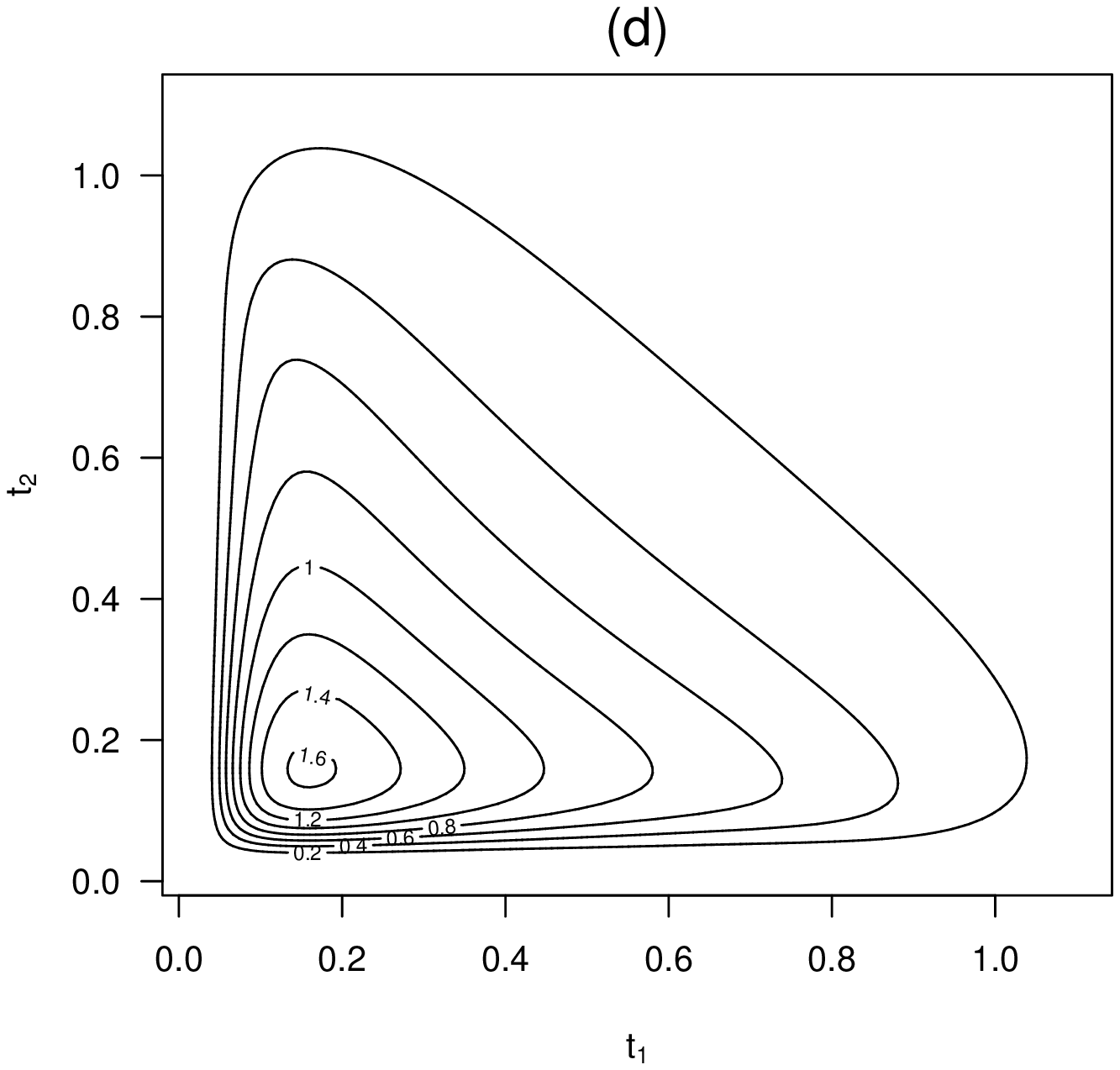}
\includegraphics[scale=0.45]{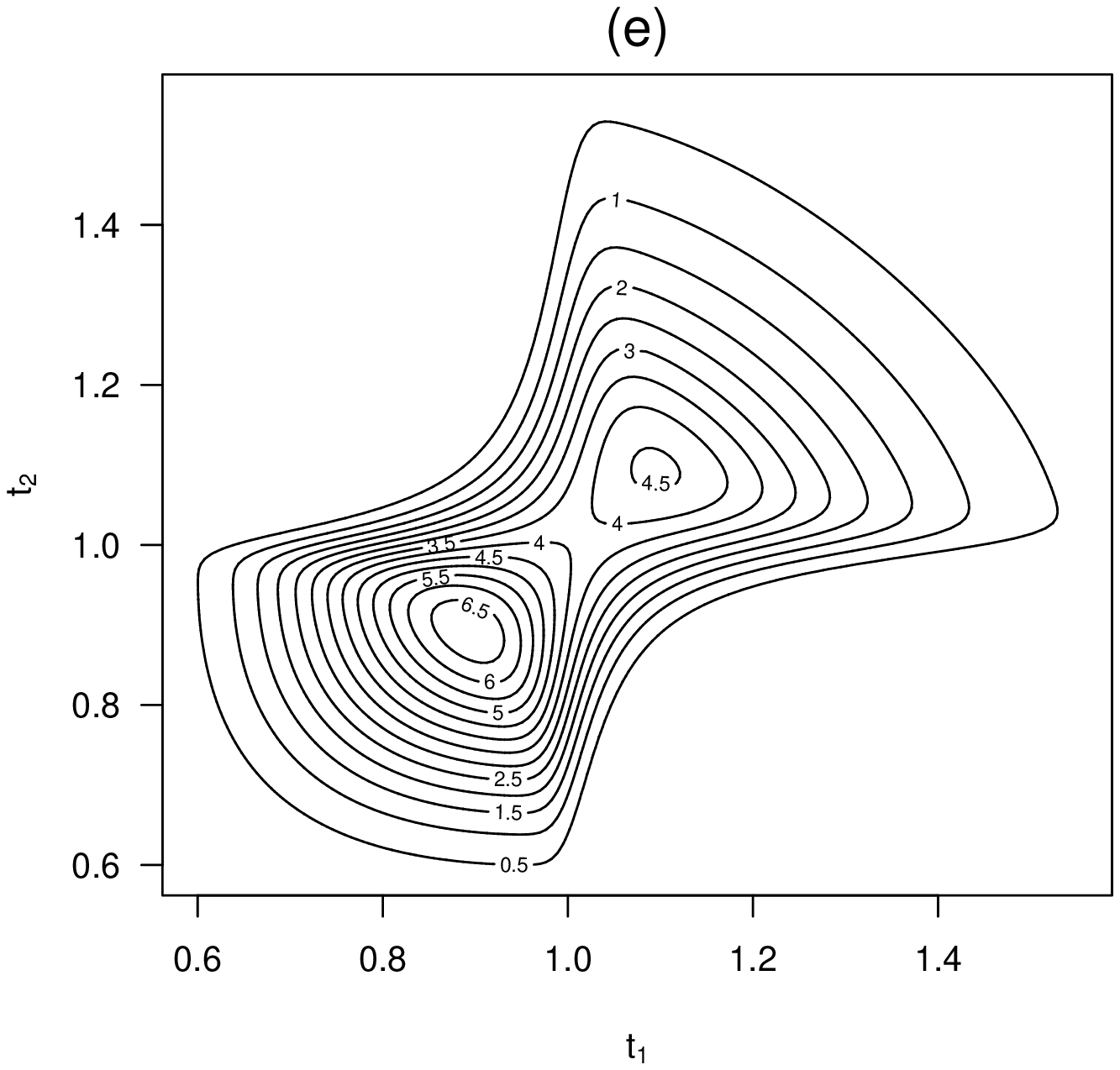}
\includegraphics[scale=0.45]{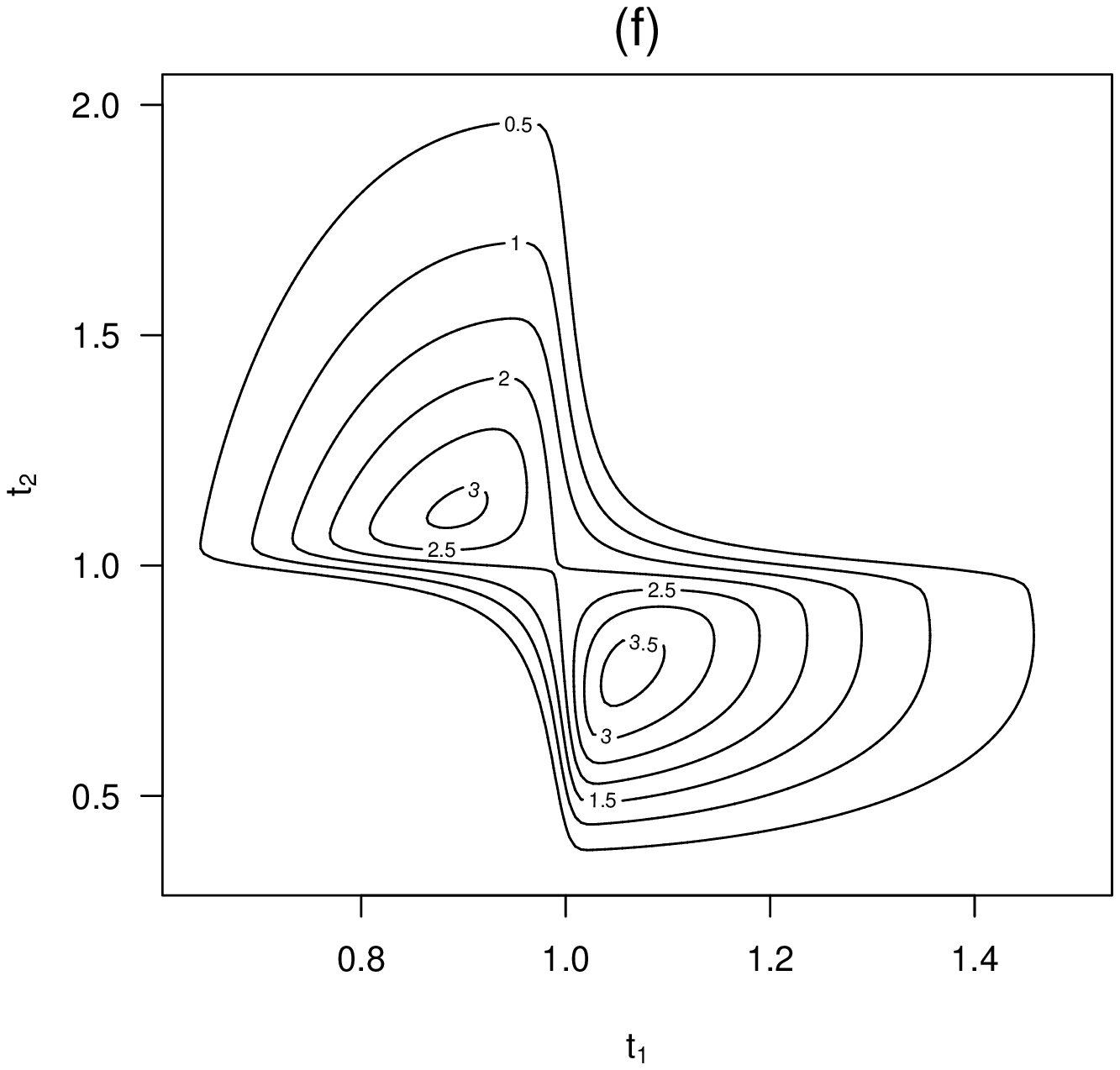}
\caption{Contour plots of the density function~(\ref{fBGBS1}) for some values of
$(\alpha_{1},\alpha_{2},\beta_{1},\beta_{2},\lambda)$: (a) $(0.5,0.5,1.0,1.0,0.5)$;
(b) $(0.2,0.4,1.0,1.0,-1)$; (c) $(0.8,0.5,1.0,1.0,1.5)$;
(d) $(1.5,1.5,1.0,1.0,1.5)$; (e) $(0.2,0.2,1.0,1.0,5)$;
(f) $(0.2,0.4,1.0,1.0,-10)$;}\label{densities_plots}
\end{figure}

The following theorem provides the marginal and conditional distributions of the SBVBS distribution.
\begin{theorem}\label{teo1}
If $(T_1,T_2) \sim\BVSBS(\alpha_{1},\alpha_{2},\beta_{1},\beta_{2},\lambda)$, then:
\begin{enumerate}
\item[{\rm (i)}]  $T_j \sim\BS(\alpha_{j},\beta_{j})$, for $j=1,2$.
\item[{\rm (ii)}] The conditional pdf of $T_1$ given $T_2=t_2$ is
\[
f_{T_1|T_2}(t_1|T_2=t_2)=2\phi(a_{1})\Phi(\lambda a_{1}a_{2})
\frac{t_1^{-3/2}(t_1+\beta_{1})}{2\alpha_{1}\sqrt{\beta_{1}}}.
\]
\item[{\rm (iii)}]  The cdf of $T_1$ given $T_2=t_2$ is
\[
\Pr(T_1\leq t_1|T_2=t_2)=\Phi(a_{1})-2\Upsilon(a_{1},\lambda a_{2}),
\]
where $\Upsilon(\cdot,\cdot)$ denotes the Owen's function \citep{Owen1956}.
\end{enumerate}
\end{theorem}
\begin{proof}
Parts (i) and (ii) follow from de definition of the distribution. We have that
\[
\Pr(T_1\leq t_1|T_2=t_2)=\int_{0}^{t_1}2\phi(a_t)\Phi(\lambda a_ta_{2})
\frac{t^{-3/2}(t+\beta_{1})}{2\alpha_{1}\sqrt{\beta_{1}}}\dd t,
\]
where
\[
a_t=\frac{1}{\alpha_{1}}\biggl[\biggl(\frac{t}{\beta_{1}}\biggr)^{1/2}
-\biggl(\frac{\beta_{1}}{t}\biggr)^{1/2}\biggr].
\]
Making the change of variable $u=a_t$, we arrive at
\[
\Pr(T_1\leq t_1|T_2=t_2)=\int_{-\infty}^{a_1}2\phi(u)\Phi(\lambda ua_{2})\dd u.
\]
Now, from \cite{Azzalini1985} we can show that
$\Pr(T_1\leq t_1|T_2=t_2)=\Phi(a_{1})-2\Upsilon(a_{1},\lambda a_{2})$ and
therefore the result (iii) holds.
\end{proof}

Some properties of the random vector $(T_1,T_2)$ are provided in the following theorem.
\begin{theorem}\label{teo2}
If $(T_1,T_2) \sim \BVSBS(\alpha_{1},\alpha_{2},\beta_{1},\beta_{2},\lambda)$, then:
\begin{enumerate}
\item[{\rm (i)}] $(k_1T_1,T_2) \sim \BVSBS(\alpha_{1},\alpha_{2},k_1\beta_{1},\beta_{2},\lambda)$,\qquad $k_1>0$.
\item[{\rm (ii)}] $(T_1,k_2T_2) \sim \BVSBS(\alpha_{1},\alpha_{2},\beta_{1},k_2\beta_{2},\lambda)$,\qquad $k_2>0$.
\item[{\rm (iii)}] $(k_1T_1,k_2T_2) \sim \BVSBS(\alpha_{1},\alpha_{2},k_1\beta_{1},k_2\beta_{2},\lambda)$,\qquad $k_1,k_2>0$.
\item[{\rm (iv)}] $(T_1^{-1},T_2^{-1}) \sim \BVSBS(\alpha_{1},\alpha_{2},\beta_{1}^{-1},\beta_{2}^{-1},\lambda)$.
\item[{\rm (v)}] $(T_1^{-1},T_2) \sim \BVSBS(\alpha_{1},\alpha_{2},\beta_{1}^{-1},\beta_{2},-\lambda)$.
\item[{\rm (vi)}] $(T_1,T_2^{-1}) \sim \BVSBS(\alpha_{1},\alpha_{2},\beta_{1},\beta_{2}^{-1},-\lambda)$.
\end{enumerate}
\end{theorem}
\begin{proof}
These follow from \eqref{fBGBS1} upon using suitable transformations.
\end{proof}

Since the marginal distributions of the bivariate vector $(T_1,T_2)$ are
BS distributions, the mean and variance of $T_1$ and $T_2$ are obtained
directly from these marginals in the forms
\[
\Es(T_j)=\beta_{j}\biggl(1+\frac{1}{2}\alpha_{j}^2\biggr),
\qquad
\Var(T_j)=(\alpha_{j}\beta_{j})^2\biggl(1+\frac{5}{4}\alpha_{j}^2\biggr),
\qquad  j = 1,2.
\]
Additionally,
\[
\Es(T_j^{-1})=\beta_{j}^{-1}\biggl(1+\frac{1}{2}\alpha_{j}^2\biggr),
\qquad
\Var(T_j^{-1})=\alpha_{j}^2\beta_{j}^{-2}\biggl(1+\frac{5}{4}\alpha_{j}^2\biggr),
\qquad  j = 1,2.
\]
The product moments of $(T_1,T_2)$, $\Es(T_1^rT_2^s)$ say,
are very complicated to be determined algebraically and have to be computed numerically.
In the following, we shall derive an expression for $\Es(T_1T_2)$ which can
be of some interest. We can show after some algebra that
\begin{align*}
\Es(T_1T_2)&=\beta_{1}\beta_{2}\Es\Biggl(\left[\frac{\alpha_{1}}{2}{Z_1}+\sqrt{\left(\frac{\alpha_{1}}{2}{Z_1}\right)^2+1}\right]^2
\left[\frac{\alpha_{2}}{2}{Z_2}+\sqrt{\left(\frac{\alpha_{2}}{2}{Z_2}\right)^2+1}\right]^2\Biggr) \\
&=\beta_{1}\beta_{2}\left[1+\frac{1}{2}(\alpha_{1}^2+\alpha_{2}^2)+\frac{1}{4}\alpha_{1}^2\alpha_{2}^2
+ 2^{1/2}\pi^{-1/2}\alpha_{1}\alpha_{2}\lambda\,\mathbb{I}\right],
\end{align*}
where
\begin{align*}
\mathbb{I}&=I_{00}+\sum_{i=2}^{\infty}u_i\,\frac{\alpha_{1}^{2i}(4i+1)!}{2^{2i}}
\sum_{m=0}^{2i}\frac{m!\,(2\lambda)^{2m}I_{im}}{(2m+1)!\,(2i-m)!} \\
&\quad+ \sum_{i=2}^{\infty}\sum_{k=2}^{\infty}{u_i\,v_k\,
\frac{\alpha_{1}^{2i}\alpha_{2}^{2k}(4i+1)!}{2^{2i}}}\sum_{m=0}^{2i}\frac{m!\,(2\lambda)^{2m}I_{ikm}}{(2m+1)!\,(2i-m)!},
\end{align*}
with
\[
u_i = (-1)^{i-1}\frac{1\times 3\times\cdots\times(2i-3)}{i!\,2^{3i}},\qquad
v_k = (-1)^{k-1}\frac{1\times 3\times\cdots\times(2k-3)}{k!\,2^{3k}},
\]
\[
I_{00}=\Es\left[\frac{1+\frac{\alpha_{2}^2Z_2^2}{2^3}}{(1+\lambda^2Z_2^2)^{1/2}}
\left(Z_2^2+\frac{30\alpha_{1}^2}{2^3(1+\lambda^2Z_2^2)^{2}}\sum_{k=1}^3c_kZ_2^{2k}\right)\right],
\]
\[
I_{ikm}=\Es\left[\frac{Z_2^{2(m+k+1)}}{(1+\lambda^2Z_2^2)^{2i+{1/2}}}\right],
\qquad
I_{im}=\Es\left[Z_2^{2(m+1)}\dfrac{1+\frac{\alpha_{2}^2Z_2^2}{2^3}}{(1+\lambda^2Z_2^2)^{2i+{1/2}}}\right],
\]
being $c_1=1$, $c_2=4\lambda^2/(3!)$ and $c_3=32\lambda^4/(5!)$. For $\lambda=0$ (independent case),
we have immediately that
\[
\Es(T_1T_2)= \beta_{1}\beta_{2}\left[1+\frac{1}{2}(\alpha_{1}^2+\alpha_{2}^2)+\frac{1}{4}\alpha_{1}^2\alpha_{2}^2\right].
\]

\section{Multivariate extension}\label{mult_case}

We have considered the bivariate case in Section \ref{bivcase}, but
extensions to higher dimension can be readily accomplished
using suitable notation. For a random variable
$\bm{Z}=(Z_1,\ldots,Z_p)^\top$ of dimension $p$, we define
the subvectors $\bm{Z}_{(1)}$,\ldots,$\bm{Z}_{(p)}$ of dimensions $(p-1)$
such that, for each $j=1,\ldots,p$, $\bm{Z}_{(j)}$ denotes the
vector $\bm{Z}$ with the $j$th coordinate $Z_{j}$ deleted.
Analogously, for a real vector $\bm{z}=(z_1,\ldots,z_p)^\top$,
$\bm{z}_{(j)}$ is obtained from  $\bm{z}$ with the $j$th coordinate $z_{j}$ deleted.

By assuming (for each $j=1,\ldots,p$) that
\begin{equation*}\label{dcex1}
Z_j|\bm{Z}_{(j)}=\bm{z}_{(j)}\sim\SN\Biggl(\lambda\prod_{j'\ne{j}}{z_{j'}}\Biggr),
\end{equation*}
the joint pdf of $\bm{Z}=(Z_1,\ldots,Z_p)^\top$
takes the form \citep{Arnold-et-al-2002}
\[
f_{\bm{Z}}(\bm{z})=2\Biggl[\prod_{j=1}^p\phi(z_{j})\Biggr]\Phi\Biggl(\lambda \prod_{j=1}^pz_{j}\Biggr),
\qquad \bm{z}\in\Reais^p.
\]
Thus, under the transformation
\begin{equation*}\label{BS3}
T_j=\beta_{j}\left[\frac{\alpha_{j}}{2}{Z_j}+\sqrt{\left(\frac{\alpha_{j}}{2}{Z_j}\right)^2+1}\right]^2,
\qquad j = 1, \ldots, p,
\end{equation*}
where $Z_j\sim\N(0,1)$, we obtain the joint pdf of $\bm{T}=(T_1,\ldots,T_p)^\top$ in the form
\begin{equation}\label{fSMVBS}
f_{\bm{T}}(\bm{t})=2\Biggl[\prod_{j=1}^p\phi(a_{j})\Biggr]\Phi\Biggl(\lambda \prod_{j=1}^pa_{j}\Biggr)
\prod_{j=1}^p\frac{t_j^{-3/2}(t_j+\beta_{j})}{2\alpha_{j}\sqrt{\beta_{j}}},  \qquad\bm{t}\in\Reais_+^p,
\end{equation}
where $\alpha_{j}>0$, $\beta_{j}>0$ and $a_{j}$ is given in \eqref{at1},
$j=1,\ldots,p$.  Let $\bm{\alpha}=(\alpha_{1},\ldots,\alpha_{p})^\top$ and
$\bm{\beta}=(\beta_{1},\ldots,\beta_{p})^\top$. If
$\bm{T}=(T_1,\ldots,T_p)^\top$ has skewed multivariate BS distribution,
then we use the notation $\bm{T}\sim\SMVBS(\bm{\alpha},\bm{\beta},\lambda)$.

Several properties discussed in the bivariate case hold for this
multivariate extension. For example, $T_{j}\sim\BS(\alpha_{j},\beta_{j})$
for $j=1,\ldots,p$, i.e.~the marginal distributions are BS distributions;
$\lambda=0$ corresponds to the independent case; for $k_1,\ldots,k_p>0$,
$(k_1T_1,\ldots,k_pT_p)\sim\SMVBS(\bm{\alpha},\bm{\beta}^*,\lambda)$
with $\bm{\beta}^*=(k_1\beta_{1},\ldots,k_p\beta_{p})^\top$;
$(T_1^{-1},\ldots,T_p^{-1})\sim\SMVBS(\bm{\alpha},\bm{\beta}^{**},\lambda)$,
where $\bm{\beta}^{**}=(\beta_{1}^{-1},\ldots,\beta_{p}^{-1})^\top$, and so on.
In the next section, we shall consider estimation for the
unknown parameters of the $\SMVBS$ distribution in
\eqref{fSMVBS} as well as inference. Thus, from these general results
the bivariate case considered in Section \ref{bivcase} can be
easily specialized by considering $p=2$.

\section{Estimation and inference}\label{estimation_infer}

In this section, we address the problem of estimating
the unknown parameters of the SMVBS distribution.
Let $\bm{t}_{1}, \ldots,\bm{t}_{n}$ denote a
random sample of the $\SMVBS(\bm{\alpha},\bm{\beta},\lambda)$
distribution, where $\bm{t}_{i}=(t_{1i},\ldots,t_{pi})^\top$
and, as before, $\bm{\alpha}=(\alpha_{1},\ldots,\alpha_{p})^\top$ and
$\bm{\beta}=(\beta_{1},\ldots,\beta_{p})^\top$. Let
$\bm{\theta}=(\bm{\alpha}^\top,\bm{\beta}^\top,\lambda)^\top$ be the
parameter vector of interest of dimension $2p+1$.

\subsection{Modified moment estimators}

First, we shall present modified moment estimators (MMEs) for the unknown parameters
by following the approach of \cite{Ng-et-al-2003}.
The $\SMVBS$ model has $2p+1$ parameters and the marginal distributions
are BS distributions with parameters $(\alpha_{j},\beta_{j})$, $j=1,\ldots,p$.
Then, the moment estimators for $\alpha_{j}$ and $\beta_{j}$
can be obtained by equating $\Es(T_j)$ and $\Var(T_j)$ to
the corresponding sample estimates for $j=1,\ldots,p$. However,
it is known that in the case of univariate BS distribution, the
moment estimators may not always exist \citep{Ng-et-al-2003}.
Here, we will use $\Es(T_j)$ and $\Es(T_j^{-1})$ instead of
using $\Es(T_j)$ and $\Var(T_j)$, and equate them to the
corresponding sample quantities. After some algebra,
the MMEs for $\alpha_{1},\ldots,\alpha_{p}$ and $\beta_{1},\ldots,\beta_{p}$ are
\begin{equation}\label{MMEs}
\check{\alpha}_j=\Biggl[2\biggl(\biggl(\frac{\bar{s}_j}{\bar{r}_j}\biggr)^{1/2}-1\biggr)\Biggr]^{1/2},
\qquad
\check{\beta}_j=(\bar{s}_j\,\bar{r}_j)^{1/2},
\qquad j=1,\ldots,p,
\end{equation}
where
\[
\bar{s}_j=\frac{1}{n}\sum_{i=1}^nt_{ji},
\qquad
\bar{r}_j=\Biggl[\frac{1}{n}\sum_{i=1}^n\frac{1}{t_{ji}}\Biggr]^{-1}.
\]
The MMEs for $\bm{\alpha}$ and $\bm{\beta}$ in \eqref{MMEs}
are explicit in form and can be used effectively as the
initial guess in the iterative process for the computation of the
maximum likelihood estimators (MLEs) in the next section.

\subsection{Maximum likelihood estimators}

The log-likelihood function for the parameter vector $\bm{\theta}$
(apart from an unimportant constant) is given by
\begin{align}  \label{logV}
\begin{split}
\ell(\bm{\theta})&=-n\sum_{j=1}^p\left[\log(\alpha_{j})+\frac{1}{2}\log(\beta_{j})\right]
+ \sum_{i=1}^n\sum_{j=1}^p\log(t_{ji}+\beta_{j})\\
&\quad -\frac{1}{2}\sum_{i=1}^n\sum_{j=1}^pa_{ji}^2+
\sum_{i=1}^n\log\Biggl[\Phi\Biggl(\lambda \prod_{j=1}^pa_{ji}\Biggr)\Biggr],
\end{split}
\end{align}
where
\[
a_{ji} =  a_{ji}(\alpha_{j},\beta_{j})=\frac{1}{\alpha_{j}}\Biggl[\biggl(\frac{t_{ji}}{\beta_{j}}\biggr)^{1/2}
-\biggl(\frac{\beta_{j}}{t_{ji}}\biggr)^{1/2}\Biggr].
\]
The MLEs of the unknown parameters are
obtained by maximizing the log-likelihood function
in \eqref{logV} with respect to $\bm{\theta}$. By taking the partial derivatives of
the log-likelihood function in \eqref{logV}
with respect to the parameters $\alpha_{j}$, $\beta_{j}$ and $\lambda$, we have
(for $j=1,\ldots,p$)
\[
\frac{\partial\ell(\bm{\theta})}{\partial\alpha_{j}}=
-\frac{n}{\alpha_{j}}+\frac{1}{\alpha_{j}}\sum_{i=1}^na_{ji}^2-\frac{\lambda}{\alpha_{j}}\sum_{i=1}^nw_i\prod_{j=1}^pa_{ji},
\]
\begin{align*}
\frac{\partial\ell(\bm{\theta})}{\partial\beta_{j}}&=-\frac{n}{2\beta_{j}}+\sum_{i=1}^n{\frac{1}{\beta_{j}+t_{ji}}}
-\frac{1}{2\alpha_{j}^2\beta_{j}}\sum_{i=1}^n\left[\frac{\beta_{j}}{t_{ji}}-\frac{t_{ji}}{\beta_{j}}\right]\\
&\quad- \frac{\lambda}{2\alpha_{j}\beta_{j}}\sum_{i=1}^nw_id_{ij} \prod_{j'\neq{j}} a_{j'i},
\end{align*}
\[
\frac{\partial\ell(\bm{\theta})}{\partial\lambda}=
\sum_{i=1}^nw_i\prod_{j=1}^pa_{ji},
\]
where
\[
w_i=w_i(\bm{\alpha},\bm{\beta},\lambda)=\frac{\phi\bigl(\lambda \prod_{j=1}^pa_{ji}\bigr)}
{\Phi\bigl(\lambda \prod_{j=1}^pa_{ji}\bigr)},
\qquad
d_{ij} = d_{ij}(\beta_{j})= \biggl(\frac{t_{ji}}{\beta_{j}}\biggr)^{1/2}+\biggl(\frac{\beta_{j}}{t_{ji}}\biggr)^{1/2}.
\]
The MLE $\widehat{\bm{\theta}} = (\widehat{\bm{\alpha}}^\top, \widehat{\bm{\beta}}^\top,\widehat{\lambda})^\top$
of $\bm{\theta}=(\bm{\alpha}^\top,\bm{\beta}^\top,\lambda)^\top$ can be obtained
by solving the likelihood equations
\begin{equation}\label{eqloglik}
\frac{\partial\ell(\bm{\theta})}{\partial\alpha_{j}}=
\frac{\partial\ell(\bm{\theta})}{\partial\beta_{j}}=
\frac{\partial\ell(\bm{\theta})}{\partial\lambda}=0,\qquad j=1,\ldots,p,
\end{equation}
simultaneously. There are no closed form expressions for the MLE
and its computation has to be performed
numerically using a nonlinear optimization algorithm.
The Newton-Raphson iterative technique could be applied to solve the
likelihood equations and obtain the estimate $\widehat{\bm{\theta}}$.
For computing the MLEs, starting values for the algorithm are required.
Since the MMEs for $\alpha_{j}$ and $\beta_{j}$ in \eqref{MMEs} are explicit,
they can be used effectively as the initial guess in the iterative procedure.
The {\tt Ox}\footnote{{\tt Ox} is freely distributed for academic purposes at
http://www.doornik.com.} matrix programming language \citep{DcK2006}
and the {\tt R} program \citep{R2010} can be used to compute $\widehat{\bm{\theta}}$
numerically.

We can show from the likelihood equations that, for given $\beta_{1},
\ldots,\beta_{p}$, the MLEs of $\alpha_{1},\ldots,\alpha_{p}$ are
\[
\widehat{\alpha}_{j}(\beta_{j}) = \biggl(\frac{\bar{s}_{j}}{\beta_{j}}
+ \frac{\beta_{j}}{\bar{r}_j} - 2\biggr)^{1/2},\qquad j=1,\ldots,p.
\]
By replacing $\alpha_{j}$ by $\widehat{\alpha}_{j}(\beta_{j})$ in \eqref{logV},
we obtain the profile log-likelihood function for $\bm{\beta}$ and
$\lambda$ as
\begin{align*}
\ell_{p}(\bm{\beta},\lambda)&=-n\sum_{j=1}^p\left[\log(\widehat{\alpha}_{j}(\beta_{j}))+\frac{1}{2}\log(\beta_{j})\right]
+ \sum_{i=1}^n\sum_{j=1}^p\log(t_{ji}+\beta_{j})\\
&\quad -\frac{1}{2}\sum_{i=1}^n\sum_{j=1}^p\widehat{a}_{ji}(\beta_{j})^2+
\sum_{i=1}^n\log\Biggl[\Phi\Biggl(\lambda \prod_{j=1}^p\widehat{a}_{ji}(\beta_{j})\Biggr)\Biggr],
\end{align*}
where
\[
\widehat{a}_{ji}(\beta_{j}) =  a_{ji}(\widehat{\alpha}_{j}(\beta_{j}),\beta_{j})
=\frac{1}{\widehat{\alpha}_{j}(\beta_{j})}\Biggl[\biggl(\frac{t_{ji}}{\beta_{j}}\biggr)^{1/2}
-\biggl(\frac{\beta_{j}}{t_{ji}}\biggr)^{1/2}\Biggr].
\]
We can also obtain the MLEs of $\bm{\beta}$ and $\lambda$ by maximizing the profile log-likelihood function
$\ell_{p}(\bm{\beta},\lambda)$ with respect to $\bm{\beta}$ and $\lambda$.
The Newton--Raphson algorithm or some
other optimization algorithm to maximize $\ell_{p}(\bm{\beta},\lambda)$
with respect to $\bm{\beta}$ and $\lambda$ needs to be used,
since the MLEs of $\bm{\beta}$ and $\lambda$ cannot be obtained explicitly.
The profile log-likelihood function $\ell_{p}(\bm{\beta},\lambda)$
is not a real log-likelihood function
and some of the properties that hold for a genuine log-likelihood do not
hold for its profiled version. In particular, there exist
score and information biases, both of order $O(1)$.

The asymptotic inference for the parameter vector
$\bm{\theta}=(\bm{\alpha}^\top,\bm{\beta}^\top,\lambda)^\top$
can be based on the normal approximation of the MLE $\widehat{\bm{\theta}}$
of $\bm{\theta}=(\bm{\alpha}^\top,\bm{\beta}^\top,\lambda)^\top$.
Under some regular conditions stated in \citet[Ch.~9]{CoxHinkley1974} that are fulfilled for
the parameters in the interior of the parameter space,  we have
$\widehat{\bm{\theta}}\stackrel{A}{\sim}\N_{2p+1}(\bm{\theta},\bm{\Sigma}_{\bm{\theta}}^{-1})$,
for $n$ large, where $\stackrel{A}{\sim}$ means approximately
distributed and $\bm{\Sigma}_{\bm{\theta}}^{-1}$ is the asymptotic variance-covariance matrix
of $\widehat{\bm{\theta}}$. The matrix $\bm{\Sigma}_{\bm{\theta}}$ is given in the Appendix.
The multivariate normal $\N_{2p+1}(\bm{0},\bm{\Sigma}_{\bm{\theta}}^{-1})$
distribution can be used to construct approximate confidence intervals for the
parameters $\alpha_{j}$, $\beta_{j}$ and $\lambda$, which are given, respectively, by
$\widehat{\alpha}_{j}\pm z_{\gamma/2}\times[\widehat{\Var}(\widehat{\alpha}_{j})]^{1/2}$,
$\widehat{\beta}_{j}\pm z_{\gamma/2}\times[\widehat{\Var}(\widehat{\beta}_{j})]^{1/2}$ and
$\widehat{\lambda}\pm z_{\gamma/2}\times[\widehat{\Var}(\widehat{\lambda})]^{1/2}$,
where $\widehat{\Var}(\cdot)$ is the diagonal element of $\bm{\Sigma}_{\bm{\theta}}^{-1}$
available at $\widehat{\bm{\theta}}$ corresponding
to each parameter, and $z_{\gamma/2}$ is the quantile $100(1-\gamma/2)\%$ of the standard normal
distribution.

Besides estimation of the model parameters, hypotheses tests can be taken into
account. Let $\bm{\theta}=(\bm{\theta}_1^\top,\bm{\theta}_2^\top)^\top$,
where $\bm{\theta}_1$ and $\bm{\theta}_2$ are disjoint subsets
of $\bm{\theta}$. Consider the test of the null
hypothesis $\mathcal{H}_0:\bm{\theta}_1=\bm{\theta}_{01}$
against $\mathcal{H}_1:\bm{\theta}_1\ne\bm{\theta}_{01}$, where $\bm{\theta}_{10}$ is
a specified vector. Let $\widetilde{\bm{\theta}}$ be the restricted MLE of $\bm{\theta}$
obtained under $\mathcal{H}_0$. The likelihood ratio (LR) statistic to test
$\mathcal{H}_0$ is given by $\omega=2\{\ell(\widehat{\bm{\theta}})-\ell(\widetilde{\bm{\theta}})\}$.
Under $\mathcal{H}_0$ and some regularity conditions,
the LR statistic converges in distribution to a chi-square distribution with dim($\bm{\theta}_1$)
degrees of freedom. In particular, the LR statistic to test the null hypothesis $\mathcal{H}_0:\lambda=0$
against $\mathcal{H}_1:\lambda\neq 0$ takes the form
\[
\omega = 2\{\ell(\widehat{\bm{\alpha}},\widehat{\bm{\beta}},\widehat{\lambda})
-\ell(\widetilde{\bm{\alpha}}, \widetilde{\bm{\beta}}, 0)\},
\]
where $\widetilde{\bm{\alpha}}$ and $\widetilde{\bm{\beta}}$ are the restricted  MLEs
of $\bm{\alpha}$ and $\bm{\beta}$, respectively, obtained from the maximization
of \eqref{logV} under $\mathcal{H}_0:\lambda=0$.
The limiting distribution of this statistic is $\chi_{1}^2$ under
the null hypothesis. The null hypothesis is rejected if the test
statistic exceeds the upper $100(1-\gamma)\%$ quantile of the $\chi_{1}^{2}$
distribution.

\section{Application to real data}\label{application}

In this section, for illustrative purposes, we present an empirical application to demonstrate
the applicability of the proposed skewed multivariate BS distribution.
For the sake of comparison, we also consider the distributions
proposed in \cite{DiazGarciaDominguezMolina2006} and \cite{Kundu-et-al-2010}.
We shall use the data set obtained from \cite{Volle1985},
which represent the amount of time (in hours)
spent on two categories of activities over 100 days in the year 1976
for 28 individuals. The data are: (115, 175), (100, 115), (130, 160),
(115, 180), (119, 143), (100, 150), (960, 132), (150, 115), (142, 870),
(180, 125), (152, 122), (174, 119), (140, 100), (147, 840), (105, 700),
(950, 600), (130, 600), (105, 800), (117, 650), (850, 400), (102, 450),
(100, 960), (920, 640), (128, 860), (102, 122), (107, 730), (860, 580),
(940, 580). The first figure represents the amount of time
spent on eating and the second figure represents the amount of time
spent on watching television.
All the computations were done using the {\tt Ox} matrix programming language \citep{DcK2006}.

We now use the SBVBS distribution to model these bivariate data.
We obtain from the data $\bar{s}_1 = 118.14$, $\bar{s}_2 = 99.43$, $\bar{r}_1 = 113.40$ and $\bar{r}_2 = 84.61$,
and hence the MMEs are $\check{\alpha}_1=0.2035$, $\check{\alpha}_2=0.4099$,
$\check{\beta}_1=115.7457$ and $\check{\beta}_2=91.7220$.
These values are used as initial guesses for $\alpha_1$, $\alpha_2$, $\beta_1$ and $\beta_2$,
respectively. An initial guess for $\lambda$ is also required to start the
maximization of the log-likelihood function \eqref{logV},
i.e.~to solve the likelihood equations \eqref{eqloglik} with $p=2$.
As initial value for $\lambda$ we consider $\check{\lambda}=0$,
which corresponds to the independent case. The algorithm converges
after 21 steps and the MLEs of $\alpha_1$, $\alpha_2$, $\beta_1$, $\beta_2$
and $\lambda$ are $\widehat{\alpha}_1=0.2047$, $\widehat{\alpha}_2=0.4101$, $\widehat{\beta}_1=113.2907$,
$\widehat{\beta}_2=90.7447$ and $\widehat{\lambda}=0.8806$, respectively. Notice that the MMEs for
$\alpha_1$, $\alpha_2$, $\beta_1$ and $\beta_2$ are close to their respective MLEs.
We have also considered other initial guesses for $\lambda$,
for example, with the initial values $\check{\lambda}=-5,-2, 3$ and 4,
the algorithm converges to the same estimates after 39, 26, 25 and 31 steps,
respectively. The 95\% asymptotic confidence intervals for $\alpha_1$, $\alpha_2$, $\beta_1$, $\beta_2$
and $\lambda$ are (0.1508, 0.2586), (0.3051, 0.5152), (104.8325, 121.7489),
(77.8975, 103.5919) and (0.0349, 1.7263), respectively.

Next, we make use of the LR statistic to test the null hypothesis $\mathcal{H}_0:\lambda=0$ against
$\mathcal{H}_1:\lambda\neq 0$. Here,
$\omega = 2\{\ell(\widehat{\alpha}_1,\widehat{\alpha}_2,\widehat{\beta}_1,\widehat{\beta}_2,\widehat{\lambda})-
\ell(\widetilde{\alpha}_1,\widetilde{\alpha}_2,\widetilde{\beta}_1,\widetilde{\beta}_2,0)\}$,
where $\widetilde{\alpha}_1$, $\widetilde{\alpha}_2$, $\widetilde{\beta}_1$ and $\widetilde{\beta}_2$
are, respectively, the restricted MLEs of $\alpha_1$, $\alpha_2$, $\beta_1$ and $\beta_2$
obtained under $\mathcal{H}_0$ and are given by
$\widetilde{\alpha}_1=0.2035$, $\widetilde{\alpha}_2=0.4099$, $\widetilde{\beta}_1=115.7470$ and
$\widetilde{\beta}_2=91.7128$. By a little computation,
we have that the LR test statistic ($\omega$) equals 6.6834\,(p-value $<0.01$). Therefore,
the null hypothesis $\mathcal{H}_0:\lambda=0$ is strongly rejected
at the usual significance levels and hence the assumption of the
skewness (correlation) is suitable for the current bivariate data.
Since the bivariate distribution in \cite{DiazGarciaDominguezMolina2006},
DG--DM say, and our proposed model are nested models
(i.e.~the DG--DM model holds for $\lambda=0$),
the null and alternative hypotheses can be rewritten as
$\mathcal{H}_0$: DG--DM against $\mathcal{H}_1$: SBVBS. Thus,
based on the LR statistic above, the SBVBS distribution
fits the data better than the bivariate DG--DM model.

The generalized LR statistic ($T_{LR,NN}$) presented in \cite{Vuong1989}
can be used for discriminating among non-nested models, which
is a distance between the two models measured in terms of the Kullback--Liebler
information criterion. Then, our proposed model and the bivariate model
in \cite{Kundu-et-al-2010} can be compared by using $T_{LR,NN}$.
For strictly nonnested models, $T_{LR,NN}$ converges in
distribution to a standard normal distribution under the null hypothesis of equivalence of the models
and the null hypothesis is not rejected if $|T_{LR,NN}|\leq z_{\gamma/2}$, where
$z_{\gamma/2}$ is the quantile $100(1-\gamma/2)\%$ of the standard normal
distribution. On the other hand, we reject at
significance level $\gamma$ the null hypothesis of equivalence of the
models in favor of the SBVBS model being better (or worse) than
the model in  \cite{Kundu-et-al-2010} if $T_{LR,NN}>z_{\gamma}$ (or $T_{LR,NN}<-z_{\gamma}$).
The generalized LR test statistic ($T_{LR,NN}$) equals $4.0903$ (p-value $<0.01$).
Therefore, the proposed model is significantly better than the model in
\cite{Kundu-et-al-2010} according to the generalized LR statistic
to model the current data.

A natural question at this point is whether SBVBS
model fits the current data satisfactorily.
Here, in order to verify it, we computed the
modified Cram\'er-von Mises ($W^*$) and Anderson-Darling ($A^*$)
statistics for the fitted marginals, i.e.~$\BS(0.2047,113.2907)$
and $\BS(0.4101,90.7447)$. The statistics $W^*$ and $A^*$ are described in details
by \cite{ChenBala1995}. The values of these statistics are
0.0971\,(p-value $>0.1$) and 0.5680\,(p-value $>0.1$), and 0.0513\,(p-value $>0.1$)
and 0.3145\,(p-value $>0.1$), respectively.
Therefore, based on the marginals, we have that the SBVBS distribution can
be used effectively in this case. Although
it does not guarantee that the bivariate real data will have SBVBS distribution, at
least it gives an indication that the SBVBS model may be used to analyze
this bivariate data.

\section{Skewed bivariate generalized BS distribution}\label{gBSbiv}

The univariate generalized BS (GBS) distribution was proposed in \cite{Diaz-Leiva05},
which is a highly flexible lifetime model that admits different degrees of kurtosis and
asymmetry and possesses unimodality and bimodality.
The GBS distribution is related to standard symmetrical distributions in $\Reais$, also
known as elliptically contoured univariate distributions. The reader is referred to
\cite{Fang-et-al-1990} and \cite{Gupta-Varga-1993} for more details about
symmetrical distributions. For the univariate case, elliptical distributions
correspond to all the symmetric distributions in $\Reais$. Specifically,
a random variable $X$ has an elliptical distribution if its probability
density function is given by $f_X(x)=c\,g([x-\mu]^2/\phi^2)$, $x\in\Reais$,
where $\mu\in\Reais$ is a location parameter and $\phi>0$ is a scale parameter.
The function $g:\Reais\to[0,\infty)$ corresponds to the kernel of the density
of $X$ and $c$ is the normalization constant such that $f_X(x)$ is a density.
The function $g(\cdot)$ is typically known as density generator.
We then write $X\sim\EL(\mu,\phi^2;g)$.

The notation $Z\sim\EL(0,1;g)$ or $Z\sim\EL(g)$ is used
for a random variable $Z$ that follows a standard elliptical distribution in
$\Reais$. The pdf and cdf of $Z$ are denoted by $f(\cdot)$ and $F(\cdot)$, respectively,
where $f(z)=c\,g(z^2)$ and $F(z)=\int_{-\infty}^zf(z)\dd z$.
The density generator of the normal, Cauchy, Student-$t$,
generalized Student-$t$, type I logistic, type II logistic
and power exponential are, respectively, given by $g(u) = (2\pi)^{-1/2}\exp(-u/2)$,
$g(u) = \{\pi(1 + u)\}^{-1}$, $g(u) = \nu^{\nu/2}B(1/2,\nu/2)^{-1}(\nu + u)^{-(\nu+1)/2}$,
where $\nu>0$ and $B(\cdot,\cdot)$ is the beta function,
$g(u) = s^{r/2}B(1/2,r/2)^{-1}(s + u)^{-(r+1)/2}$ ($s,r>0$),
$g(u) = c\,\e^{-u}(1+\e^{-u})^{-2}$, where $c\approx 1.484300029$ is the
normalizing constant obtained from $\int_{0}^{\infty}u^{-1/2}g(u)\dd u=1$,
$g(u) = \e^{-\sqrt{u}}(1+\e^{-\sqrt{u}})^{-2}$ and
$g(u) = c(k)\exp(-\frac{1}{2}u^{1/(1+k)})$, $-1<k\leq 1$,
where $c(k)=\Gamma(1+(k+1)/2)2^{1+(1+k)/2}$.

In the following, we shall introduce the skewed bivariate GBS (SBVGBS) distribution.
A random variable $Y$ follows a standard skew-elliptical distribution in
$\Reais$ if its pdf takes the form
\begin{equation}\label{fdpskeweliptica}
f_Y(y)=2f(y)F(\lambda y),\qquad y\in\Reais.
\end{equation}
We use the notation $Y\sim\SE(\lambda;g)$. If $\lambda=0$ in \eqref{fdpskeweliptica},
then the standard elliptical distribution holds, i.e.~$Y\sim\EL(g)$.
Now, let $Z_j \sim\EL(g)$, for $j=1,2$, with $Z_1|Z_2=z_2\sim\SE(\lambda z_2;g)$ and
$Z_2|Z_1=z_1\sim\SE(\lambda z_{1};g)$. Additionally, consider the transformation
\[
T_j=\beta_{j}\left[\frac{\alpha_{j}}{2}{Z_j}
+\sqrt{\left(\frac{\alpha_{j}}{2}{Z_j}\right)^2+1}\right]^2,\qquad j=1,2,
\]
where $\alpha_{j}>0$ and $\beta_{j}>0$. Then, from the above transformation
and using results due to \cite{Arnold-et-al-2002}, the joint pdf of the SBVGBS
distribution is given by
\begin{equation}\label{fBGBS2}
f_{T_1,T_2}(t_1,t_2)=2f(a_{1})f(a_{2})F(\lambda a_{1}a_{2})
\frac{t_1^{-3/2}(t_1+\beta_{1})}{2\alpha_{1}\sqrt{\beta_{1}}}
\frac{t_2^{-3/2}(t_2+\beta_{2})}{2\alpha_{2}\sqrt{\beta_{2}}},\qquad (t_1,t_2)\in\Reais_+^2,
\end{equation}
where $a_{j}$ is defined in \eqref{at1}.
If $(T_1,T_2)$ follows the SBVGBS distribution,
the notation used is $(T_1, T_2)\sim\SBVGBS(\alpha_{1}, \alpha_{2}, \beta_{1}, \beta_{2}, \lambda;g)$.
Notice that the joint pdf \eqref{fBGBS1} is a special case of \eqref{fBGBS2}.
All extra parameters are considered as known or fixed in \eqref{fBGBS2}. For example,
the degrees of freedom for the Student-$t$ model. The main motivation
for this generalization of the SBVBS model presented in Section \ref{bivcase}
is based on the search for bivariate distributions that are more
flexible than the SBVBS model in analyzing bivariate data.

Some properties for this bivariate class of distributions
are presented in the following theorem.
\begin{theorem}\label{teo3}
If $(T_1,T_2) \sim\SBVGBS(\alpha_{1},\alpha_{2},\beta_{1},\beta_{2},\lambda;g)$, then:
\begin{enumerate}
\item[{\rm (i)}] $T_j \sim\GBS(\alpha_{j},\beta_{j};g)$, for $j=1,2$.
\item[{\rm (ii)}] $(k_1T_1,T_2) \sim\SBVGBS(\alpha_{1},\alpha_{2},k_1\beta_{1},\beta_{2},\lambda;g)$,\qquad $k_1>0$.
\item[{\rm (iii)}] $(T_1,k_2T_2) \sim\SBVGBS(\alpha_{1},\alpha_{2},\beta_{1},k_2\beta_{2},\lambda;g)$,\qquad $k_2>0$.
\item[{\rm (iv)}] $(k_1T_1,k_2T_2) \sim\SBVGBS(\alpha_{1},\alpha_{2},k_1\beta_{1},k_2\beta_{2},\lambda;g)$,\qquad $k_1,k_2>0$.
\item[{\rm (v)}] $(T_1^{-1},T_2^{-1}) \sim\SBVGBS(\alpha_{1},\alpha_{2},\beta_{1}^{-1},\beta_{2}^{-1},\lambda;g)$.
\item[{\rm (vi)}] $(T_1^{-1},T_2) \sim\SBVGBS(\alpha_{1},\alpha_{2},\beta_{1}^{-1},\beta_{2},-\lambda;g)$.
\item[{\rm (vii)}] $(T_1,T_2^{-1}) \sim\SBVGBS(\alpha_{1},\alpha_{2},\beta_{1},\beta_{2}^{-1},-\lambda;g)$.
\end{enumerate}
\end{theorem}
\begin{proof}
Using suitable transformations in \eqref{fBGBS2}, these results follow.
\end{proof}

From \eqref{fBGBS2}, several news SBVGBS distributions can be obtained.
For example, the joint pdf of the skewed bivariate BS Student-$t$ model takes the form
\[
f_{T_1,T_2}(t_1,t_2)=2\prod_{j=1}^2\frac{\Gamma([\nu_{j}+1]/2)}{(\nu_{j}\pi)^{1/2}\Gamma(\nu_{j}/2)}
\biggl(1+\frac{a_{j}^2}{\nu_{j}}\biggr)^{\frac{\nu_{j} + 1}{2}}
\frac{t_{j}^{-3/2}(t_{j}+\beta_{j})}{2\alpha_{j}\sqrt{\beta_{j}}}
\frac{1}{2}\Biggl[1 + I_{q_j}\biggl(\frac{1}{2},\frac{\nu_{j}}{2}\biggr)\Biggr],
\]
where $\nu_{j}$ is the degrees of freedom,
$q_j=(\lambda a_1a_2)^2/[(\lambda a_1a_2)^2 + \nu_{j}]$ and  $I_x(r,s)$
is the incomplete beta ration function. The skewed bivariate BS
Cauchy distribution is a special case of the joint pdf above
when $\nu_{1}=\nu_{2}=1$. It is evident that
other bivariate models can be obtained as, for example,
the skewed bivariate BS type I (type II) logistic model,
skewed bivariate BS power exponential model, and so on.
Further, extensions to higher dimension can be derived and
MLE of the unknown parameters can also be considered.
These problems can be developed in a future research.

\section{Concluding remarks}\label{conclusions}

The univariate BS model has many attractive properties and has found several
applications in the literature including lifetime, survival and environmental
data analysis \citep[see, for example,][]{Leiva-et-al-2008,Leiva-et-al-2009}.
As mentioned before, little work has been done to extend the BS model
for the multivariate case. In this article, we have introduced the skewed
multivariate BS distribution. The new distribution
is very general, quite flexible and widely applicable.
The new model is an absolutely continuous multivariate distribution whose marginals are univariate
BS distributions. We have discussed several properties of this new class
of distributions and the estimation of parameters is
approached by the method of maximum likelihood. The observed and
expected information matrices are determined and likelihood ratio tests for
some hypotheses of interest are also considered.
The skewed bivariate BS distribution is discussed
and we have shown that the additional shape parameter ($\lambda$) introduces
skewness, correlation and bimodality to this distribution.
These interesting properties make this bivariate model
a quite flexible distribution to model bivariate data. Other bivariate BS models
have been introduced and are given in \cite{DiazGarciaDominguezMolina2006}
and \cite{Kundu-et-al-2010}, KBJ say.
The DG--DM model is an independent bivariate model and hence does not
consider correlation between the random bivariate vector. The KBJ model
considers correlation between the random bivariate vector, but
does not allow bimodality. As remarked, the skewed bivariate
BS model proposed in this article can be skewed, correlated and bimodal, and
therefore is much more flexible than the other bivariate BS models
available in the literature for analyzing bivariate data. This is supported in an application
to real data in which we show that the skewed bivariate
BS model provides consistently better fit than the DG--DM and KBJ models.
Finally, we have also introduced in this paper the skewed bivariate generalized BS distribution
and discussed some of its properties. Although we have discussed
the generalized BS distribution in bivariate settings,
the skewed multivariate generalized
BS distribution can be introduced along the same lines.
This problem can be developed in a future research.

\section*{Acknowledgments}

We gratefully acknowledge grants from FAPESP (Brazil)
and Mobility Program of the Universidad Industrial de Santander
(Colombia).

\appendix

{
\section*{Appendix. Fisher information matrix}\label{FIM_theta}

We present the elements of the Fisher information matrix $\bm{\Sigma}_{\bm{\theta}}$.
First, we shall compute the elements of the Hessian
matrix
\[
\Ddot{\bm{L}}_{\bm{\theta}\bm{\theta}} = \frac{\partial^2\ell(\bm{\theta})}{\partial\bm{\theta}\partial\bm{\theta}^\top}=
\begin{bmatrix}
\Ddot{\bm{L}}_{\bm{\alpha}\bm{\alpha}} & \Ddot{\bm{L}}_{\bm{\alpha}\bm{\beta}} &  \Ddot{\bm{L}}_{\bm{\alpha}\lambda} \\
\Ddot{\bm{L}}_{\bm{\alpha}\bm{\beta}}^\top & \Ddot{\bm{L}}_{\bm{\beta}\bm{\beta}} & \Ddot{\bm{L}}_{\bm{\beta}\lambda} \\
\Ddot{\bm{L}}_{\bm{\alpha}\lambda}^\top & \Ddot{\bm{L}}_{\bm{\beta}\lambda}^\top & \Ddot{L}_{\lambda\lambda}
\end{bmatrix},
\]
with
\[
\Ddot{\bm{L}}_{\bm{\alpha}\bm{\alpha}} = \frac{\partial^2\ell(\bm{\theta})}{\partial\bm{\alpha}\partial\bm{\alpha}^\top}
=((\Ddot{L}_{\alpha_{j'}\alpha_{j}})), \qquad
\Ddot{\bm{L}}_{\bm{\alpha}\bm{\beta}} = \frac{\partial^2\ell(\bm{\theta})}{\partial\bm{\alpha}\partial\bm{\beta}^\top}
=((\Ddot{L}_{\alpha_{j'}\beta_{j}})),
\]
\[
\Ddot{\bm{L}}_{\bm{\alpha}\lambda} = \frac{\partial^2\ell(\bm{\theta})}{\partial\bm{\alpha}\partial\lambda}
=(\Ddot{L}_{\alpha_{1}\lambda}, \ldots, \Ddot{L}_{\alpha_{p}\lambda})^\top, \qquad
\Ddot{\bm{L}}_{\bm{\beta}\bm{\beta}} = \frac{\partial^2\ell(\bm{\theta})}{\partial\bm{\beta}\partial\bm{\beta}^\top}
=((\Ddot{L}_{\beta_{j'}\beta_{j}})),
\]
\[
\Ddot{\bm{L}}_{\bm{\beta}\lambda} = \frac{\partial^2\ell(\bm{\theta})}{\partial\bm{\beta}\partial\lambda}
=(\Ddot{L}_{\beta_{1}\lambda},\ldots,\Ddot{L}_{\beta_{p}\lambda})^\top, \qquad
\Ddot{L}_{\lambda\lambda} = \frac{\partial^2\ell(\bm{\theta})}{\partial\lambda^2},
\]
where  $j,j'=1,\ldots,p$,
\begin{align*}
\Ddot{L}_{\alpha_{j}\alpha_{j}}&=\frac{n}{\alpha_{j}^2}-\frac{3}{\alpha_{j}^2}\sum_{i=1}^na_{ji}^2
+\frac{2\lambda}{\alpha_{j}^2}\sum_{i=1}^{n}w_i\prod_{j=1}^{p}a_{ji}\\
&\quad-\frac{\lambda^3}{\alpha_{j}^2}\sum_{i=1}^{n}w_i\prod_{j=1}^{p}a_{ji}^3
-\frac{\lambda^2}{\alpha_{j}^2}\sum_{i=1}^{n}w_i^2\prod_{j=1}^{p}a_{ji}^2,
\end{align*}
\begin{align*}
\Ddot{L}_{\alpha_{j'}\alpha_{j}}&=\frac{\lambda}{\alpha_{j}\alpha_{j'}}\sum_{i=1}^{n}w_i\prod_{j=1}^{p}a_{ji}
-\frac{\lambda^2}{\alpha_{j}\alpha_{j'}}\sum_{i=1}^{n}w_i^2\prod_{j=1}^{p}a_{ji}^2 \\
&\quad -\frac{\lambda^3}{\alpha_{j}\alpha_{j'}}\sum_{i=1}^{n}w_i\prod_{j=1}^{p}a_{ji}^3,\qquad j'\neq j,
\end{align*}
\begin{align*}
\Ddot{L}_{\alpha_{j}\beta_{j}}&=\frac{1}{\alpha_{j}^3\beta_{j}}\sum_{i=1}^n\biggl[\frac{\beta_{j}}{t_{ji}}-\frac{t_{ji}}{\beta_{j}}\biggr]
+\frac{\lambda}{2\alpha_{j}^2\beta_{j}}\sum_{i=1}^{n}w_id_{ij}\prod_{j'\neq j}a_{j'i}\\
&\quad-\frac{\lambda^2}{2\alpha_{j}^2\beta_{j}}\sum_{i=1}^{n}w_id_{ij}
\Biggl[\lambda\prod_{j=1}^{p}a_{ji}^2+w_i\prod_{j=1}^{p}a_{ji}\Biggr]\prod_{j'\neq j}a_{j'i},
\end{align*}
\begin{align*}
\Ddot{L}_{\alpha_{j'}\beta_{j}}&=-\frac{\lambda}{2\alpha_{j}\alpha_{j'}\beta_{j}}\sum_{i=1}^nw_id_{ij}
\Biggl[-1 + \lambda^2\prod_{j=1}^{p}a_{ji}^2+\lambda w_i\prod_{j=1}^{p}a_{ji}\Biggr]\prod_{j'\neq j}a_{j'i}, \qquad j'\neq j,
\end{align*}
\begin{align*}
\Ddot{L}_{\beta_{j}\beta_{j}}&=\frac{n}{2\beta_{j}^2}-\sum_{i=1}^n(t_{ji}+\beta_{j})^{-2}-\frac{1}{\alpha_{j}^2\beta_{j}^3}\sum_{i=1}^nt_{ji}
+\frac{\lambda}{4\beta_{j}^2}\sum_{i=1}^{n}w_i\prod_{j=1}^{p}a_{ji}\\
&\quad-\frac{\lambda}{2\alpha_{j}\beta_{j}^2}\sum_{i=1}^nw_id_{ij}
\Biggl[-1 + \frac{\lambda}{2\alpha_{j}}d_{ij}\Biggl(w_i + \lambda\prod_{j=1}^{p}a_{ji}\Biggl)\prod_{j'\neq j}a_{j'i}\Biggr]
\prod_{j'\neq j}a_{j'i},
\end{align*}
\begin{align*}
\Ddot{L}_{\beta_{j'}\beta_{j}}&=-\frac{\lambda}{4\alpha_{j}\beta_{j}\alpha_{j'}\beta_{j'}}\sum_{i=1}^nw_id_{ij}d_{ij'}
\Biggl[-1 + \lambda\Biggl(w_i + \lambda\prod_{j=1}^{p}a_{ji}\Biggl)\prod_{j'\neq j}a_{j'i}\Biggr]\prod_{j'\neq j}a_{j'i},\qquad j'\neq j,
\end{align*}
\[
\Ddot{L}_{\alpha_{j}\lambda}=-\frac{1}{\alpha_{j}}\sum_{i=1}^nw_i
\Biggl[1 - \lambda\Biggl(\lambda\prod_{j=1}^{p}a_{ji}^2+w_i\prod_{j=1}^{p}a_{ji}\Biggl)\Biggr]\prod_{j=1}^pa_{ji},
\]
\[
\Ddot{L}_{\beta_{j}\lambda}=-\frac{1}{2\alpha_{j}\beta_{j}}\sum_{i=1}^nw_id_{ij}
\Biggl[1 - \lambda\Biggl(\lambda\prod_{j=1}^{p}a_{ji}^2+w_i\prod_{j=1}^{p}a_{ji}\Biggl)\Biggr]\prod_{j'\neq j}a_{j'i},
\]
\[
\Ddot{L}_{\lambda\lambda}=-\sum_{i=1}^nw_i
\Biggl(\lambda\prod_{j=1}^{p}a_{ji}^2+w_i\prod_{j=1}^{p}a_{ji}\Biggl)\prod_{j=1}^{p}a_{ji}.
\]

The Fisher information matrix is given by
\[
\bm{\Sigma}_{\bm{\theta}} =  -\Es(\Ddot{\bm{L}}_{\bm{\theta}\bm{\theta}})=
\begin{bmatrix}
\bm{\Sigma}_{\bm{\alpha}\bm{\alpha}} & \bm{\Sigma}_{\bm{\alpha}\bm{\beta}} &  \bm{\Sigma}_{\bm{\alpha}\lambda} \\
\bm{\Sigma}_{\bm{\alpha}\bm{\beta}}^\top & \bm{\Sigma}_{\bm{\beta}\bm{\beta}} & \bm{\Sigma}_{\bm{\beta}\lambda} \\
\bm{\Sigma}_{\bm{\alpha}\lambda}^\top & \bm{\Sigma}_{\bm{\beta}\lambda}^\top & \Sigma_{\lambda\lambda}
\end{bmatrix},
\]
where
\[
{\bm{\Sigma}}_{\bm{\alpha}\bm{\alpha}} =((\Sigma_{\alpha_{j'}\alpha_{j}})), \qquad
{\bm{\Sigma}}_{\bm{\alpha}\bm{\beta}} =((\Sigma_{\alpha_{j'}\beta_{j}})),
\qquad
{\bm{\Sigma}}_{\bm{\alpha}\lambda} =(\Sigma_{\alpha_{1}\lambda},\ldots,\Sigma_{\alpha_{p}\lambda})^\top,
\]
\[
{\bm{\Sigma}}_{\bm{\beta}\bm{\beta}} =((\Sigma_{\beta_{j'}\beta_{j}})),
\qquad
{\bm{\Sigma}}_{\bm{\beta}\lambda} =(\Sigma_{\beta_{1}\lambda},\ldots,\Sigma_{\beta_{p}\lambda})^\top,
\]
for  $j,j'=1,\ldots,p$,
\[
\Sigma_{\alpha_{j}\alpha_{j}}=\frac{2n}{\alpha_{j}^2}+\frac{\lambda^3}{\alpha_{j}^2}\sum_{i=1}^{n}\Es\Biggl[w_i\prod_{j=1}^{p}a_{ji}^3\Biggr]
+\frac{\lambda^2}{\alpha_{j}^2}\sum_{i=1}^{n}\Es\Biggl[w_i^2\prod_{j=1}^{p}a_{ji}^2\Biggr],
\]
\[
\Sigma_{\alpha_{j'}\alpha_{j}}=\frac{\lambda^3}{\alpha_{j}\alpha_{j'}}\sum_{i=1}^{n}\Es\Biggl[w_i\prod_{j=1}^{p}a_{ji}^3\Biggr]
+\frac{\lambda^2}{\alpha_{j}\alpha_{j'}}\sum_{i=1}^{n}\Es\Biggl[w_i^2\prod_{j=1}^{p}a_{ji}^2\Biggr],\qquad j'\neq j,
\]
\[
\Sigma_{\alpha_{j}\beta_{j}}=\frac{\lambda^2}{2\alpha_{j}^2\beta_{j}}\sum_{i=1}^{n}\Es\Biggl[w_id_{ij}
\Biggl(\lambda\prod_{j=1}^{p}a_{ji}^2+w_i\prod_{j=1}^{p}a_{ji}\Biggr)\prod_{j'\neq j}a_{j'i}\Biggr],
\]
\[
\Sigma_{\alpha_{j'}\beta_{j}}=\frac{\lambda^2}{2\alpha_{j}\alpha_{j'}\beta_{j}}\sum_{i=1}^{n}\Es\Biggl[w_id_{ij}
\Biggl(\lambda\prod_{j=1}^{p}a_{ji}^2+w_i\prod_{j=1}^{p}a_{ji}\Biggr)\prod_{j'\neq j}a_{j'i}\Biggr], \qquad j'\neq j,
\]
\[
\Sigma_{\beta_{j}\beta_{j}}=\frac{n}{\alpha_{j}^2\beta_{j}^2}+\frac{nK(\alpha_{j})}{\alpha_{j}\beta_{j}^2}
+\frac{\lambda^2}{4\alpha_{j}^2\beta_{j}^2}\sum_{i=1}^n\Es\Biggl[w_id_{ij}^2
\Biggl(w_i + \lambda w_i\prod_{j=1}^{p}a_{ji}\Biggl)\prod_{j'\neq j}a_{j'i}^2\Biggr],
\]
\begin{align*}
\Sigma_{\beta_{j'}\beta_{j}}&=\frac{\lambda}{4\alpha_{j}\beta_{j}\alpha_{j'}\beta_{j'}}\sum_{i=1}^n
\Es\Biggl\{w_id_{ij}d_{ij'}\Biggl[-1 \\
&\qquad\qquad+ \lambda\Biggl(w_i + \lambda\prod_{j=1}^{p}a_{ji}\Biggl)\prod_{j'\neq j}a_{j'i}\Biggr]
\prod_{j'\neq j}a_{j'i}\Biggr\},\qquad j'\neq j,
\end{align*}
\[
\Sigma_{\alpha_{j}\lambda}=-\frac{\lambda}{\alpha_{j}}\sum_{i=1}^n\Es\Biggl[w_i
\Biggl(\lambda\prod_{j=1}^{p}a_{ji}^2+w_i\prod_{j=1}^{p}a_{ji}\Biggl)\prod_{j=1}^pa_{ji}\Biggr],
\]
\[
\Sigma_{\beta_{j}\lambda}=-\frac{\lambda}{2\alpha_{j}\beta_{j}}\sum_{i=1}^n\Es\Biggl[w_id_{ij}
\Biggl(\lambda\prod_{j=1}^{p}a_{ji}^2+w_i\prod_{j=1}^{p}a_{ji}\Biggl)\prod_{j'\neq j}a_{j'i}\Biggr],
\]
\[
\Sigma_{\lambda\lambda}=\sum_{i=1}^n\Es\Biggl[w_i
\Biggl(\lambda\prod_{j=1}^{p}a_{ji}^2+w_i\prod_{j=1}^{p}a_{ji}\Biggl)\prod_{j=1}^{p}a_{ji}\Biggr].
\]
All the expected values above are obtained numerically.
Also, $K(\alpha_{j})=\bigl[\alpha_{j} - \sqrt{\pi}K^*(\alpha_{j})/\sqrt{2}\bigr]/2$,
with $K^*(\alpha_{j})=[1 - \mathtt{erf}(\sqrt{2}/\alpha_{j})]\exp(2/\alpha_{j}^2)$,
for $j=1,\ldots,p$, where ${\tt erf}(\cdot)$ is the error function given by
$\mathtt{erf}(x)=(2/\sqrt{\pi})\int_{0}^{x}\mathrm{e}^{-t^2}\mathrm{d}t$.
Details on $\mathtt{erf}(\cdot)$ can be found in \citet{GradshteynRyzhik2007}.
For small values of $\alpha$ \citep[p.~298]{AE70}
\begin{equation}\label{a0approx}
K^*(\alpha_{j}) \approx \frac{\alpha_{j}}{\sqrt{2\pi}}\biggl(1 - \frac{\alpha_{j}^2}{4} +
\frac{3\alpha_{j}^4}{16}\biggr).
\end{equation}
For numerical evaluation we recommend the use of~(\ref{a0approx}) when $\alpha < 0.5$.

For $\lambda=0$, which corresponds to the
independent case, we obtain the Fisher information matrix
\[
\bm{\Sigma}_{\bm{\theta}\bm{\theta}} = n\,
\mbox{block-diag}\{\bm{\Sigma}_{\bm{\alpha}\bm{\alpha}}, \bm{\Sigma}_{\bm{\beta}\bm{\beta}}, 2/\pi\},
\]
where $\bm{\Sigma}_{\bm{\alpha}\bm{\alpha}}=2\,\diag\{\alpha_{1}^{-2},\ldots,\alpha_{p}^{-2}\}$,
$\bm{\Sigma}_{\bm{\beta}\bm{\beta}}=\diag\{b_{1},\ldots,b_{p}\}$,
with $b_{j}=[\alpha_{j}K(\alpha_{j}) + 1]/(\alpha_{j}^2\beta_{j}^2)$ for $j=1,\ldots,p$.
It can be shown that
\[
|\bm{\Sigma}_{\bm{\theta}\bm{\theta}}|=\frac{2^{p+1}n^{2p+1}}{\pi}
\prod_{j=1}^p\frac{[\alpha_{j}K(\alpha_{j}) + 1]}{\alpha_{j}^4\beta_{j}^2}\neq{0}.
\]
Therefore, the Fisher information matrix is not singular at $\lambda=0$.

Finally, it is well known that under some mild regularity conditions,
the asymptotic behavior remains valid if $\bm{\Sigma}_{\bm{\theta}}$ is
approximated by $-\Ddot{\bm{L}}_{\widehat{\bm{\theta}}\widehat{\bm{\theta}}}$, where
$-\Ddot{\bm{L}}_{\widehat{\bm{\theta}}\widehat{\bm{\theta}}}$ is the $(2p+1)\times(2p+1)$
observed information matrix evaluated at $\widehat{\bm{\theta}}$, obtained from
$\Ddot{\bm{L}}_{\bm{\theta}\bm{\theta}}$.
So, in order to avoid numerical integrations, one can use
$-\Ddot{\bm{L}}_{\widehat{\bm{\theta}}\widehat{\bm{\theta}}}$
instead of $\bm{\Sigma}_{\bm{\theta}}$ to make inference.
}


\end{document}